\documentclass[letterpaper, 10 pt, conference]{ieeeconf}  

\IEEEoverridecommandlockouts                              
\usepackage{cite}
\usepackage{url}
\usepackage{bm}
\usepackage{graphicx}
\usepackage{algorithm}
\usepackage[noend]{algorithmic}
\usepackage{subfig}
\usepackage{comment}
\usepackage{amssymb, amsmath,graphicx,charter, latexsym}
\usepackage{enumerate}

\newtheorem{lemma}{Lemma}

\newtheorem{theorem}{Theorem}

\allowdisplaybreaks[4]
\usepackage{bbm}

\title{\bf{Pathwise Performance of Debt Based Policies for Wireless Networks with Hard Delay Constraints
}}

\author{Rahul Singh, I-Hong Hou, and P. R. Kumar
  \thanks{This material is based upon work partially supported by AFOSR under contract No. FA 9550-13-1-0008, the NSF under the Science and Technology Center Grant CCF-0939370, and contract Nos. CNS-1232602, CNS-1232601, CNS-1239116, and USARO under Contract No. W911NF-08-1-0238.}
   \\Department of Electrical and Computer Engineering, Texas A\&M University
   \\Email: \{rsingh1,ihou,prk\}@tamu.edu
}

\begin{document}

\maketitle
\begin{abstract}
Hou et al have introduced  a framework to serve clients over wireless channels when there are hard deadline constraints along with a minimum delivery ratio for each client's flow. Policies based on ``debt," called maximum debt first
policies (MDF) were introduced, and shown to be throughput optimal. By ``throughput optimality" it is meant
that if there exists a policy that fulfils a set of clients with a given vector of delivery ratios and a vector of channel reliabilities, then the MDF policy will also fulfill them. The debt of a user is the difference between the number of packets that should have been delivered so as to meet the delivery ratio and the number of packets that have been delivered for that client. The maximum debt first (MDF) prioritizes the clients in decreasing order of debts at the beginning of every period. 
Note that a throughput optimal policy only guarantees that
\begin{small} 
$\liminf_{T \to \infty} \frac{1}{T}\sum_{t=1}^{T} \mathbbm{1}\{ \mbox{client $n$'s packet is delivered in frame $t$} \}\\ \geq q_{i}$
\end{small}, where the right hand side is the required delivery ratio for client $i$. Thus, it only guarantees that the debts of each user are $o(T)$, and can be otherwise arbitrarily large. This raises the interesting question about what is the growth rate of the debts under the MDF policy.
We show the optimality of MDF policy in the case when the channel reliabilities of all users are same, and obtain performance bounds for the general case. For the performance bound we obtain the almost sure bounds on $\limsup_{t\to\infty}\frac{d_{i}(t) }{\phi(t)}$ for all $i$, where $\phi(t) = \sqrt{2t\log\log t}$.
\end{abstract}
\section{Introduction}\label{sec1}
Consider a wireless network consisting of an access point serving $N$ clients as shown in Figure~\ref{fig1}. We will assume that time is divided into slots. Let us call the set of slots $k\tau,\ldots,\left(k+1\right)\tau-1$ as the $k$-th frame, and $\tau$ as the period. At the beginning of each frame, the access point generates one packet for each of the $N$ clients. If such a packet is to be useful, it should be
delivered in the same frame. That is, if a packet generated at the beginning of slot $k\tau$ is not delivered to the client by slot $\left(k+1\right)\tau-1$, then it is considered to have expired and is dropped. 
We shall call the throughput of packets per frame that are delivered for client $i$ as its timely throughput. We shall suppose that each client $i$ requires a timely throughput $q_{i}$.
\begin{figure}[b]
\includegraphics[scale=0.25]{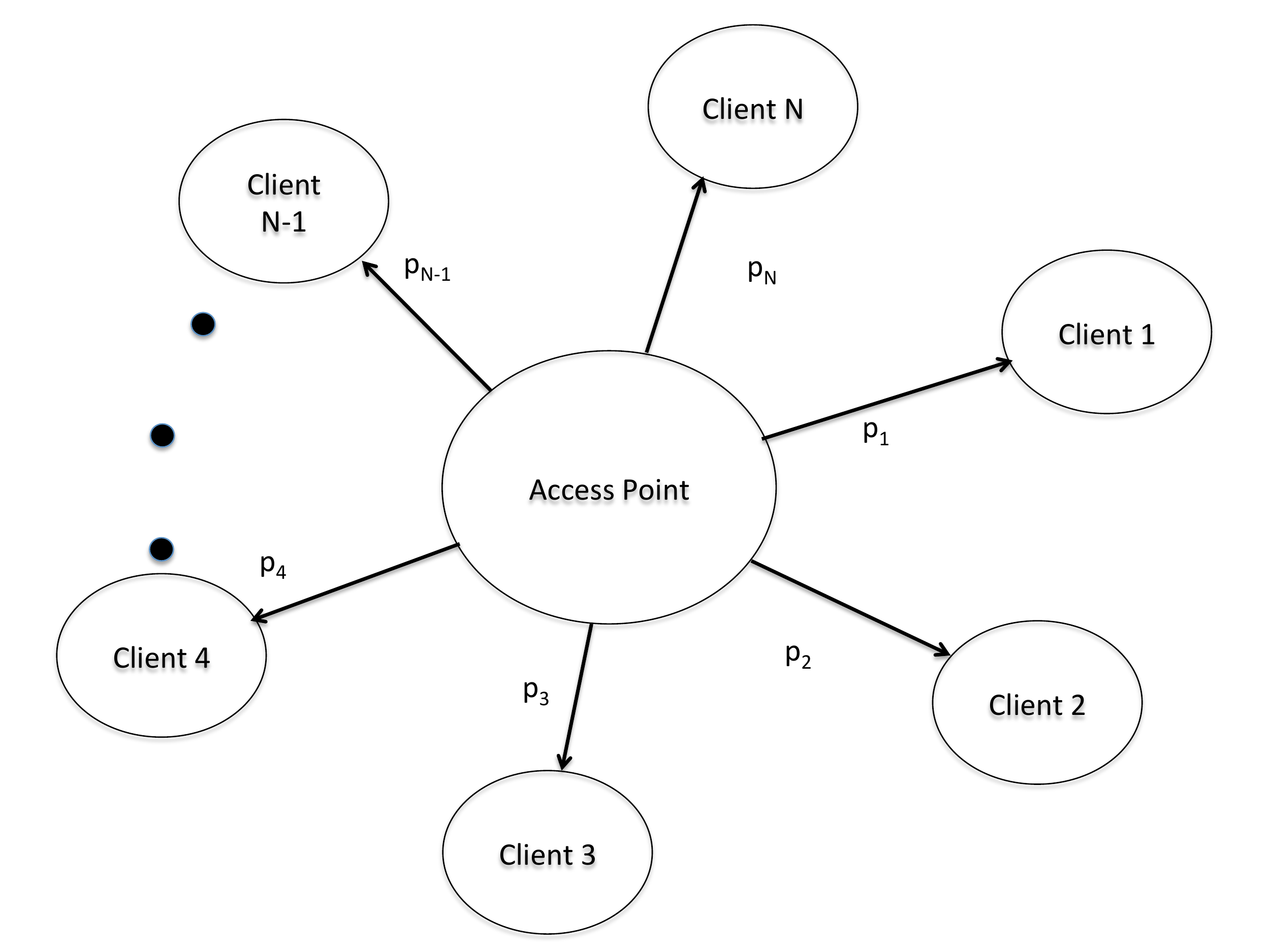}
\caption{An access point serving $N$ real-time flows.}
\label{fig1}
\end{figure}
In each slot, the access point can transmit one packet. So it has to choose which client's packet to transmit, from among those clients that still have undelivered packets in that frame. The channels between the access point and the clients are however not reliable. A packet transmitted to client $i$ is successfully delivered in that slot with probability $p_{i}$, and fails with probability $\left(1-p_{i}\right)$. The access point can retransmit a failed packet at a later slot in the frame, before it expires.\\
This model of a wireless network serving real-time flows with per packet hard deadlines was introduced in Hou et al~\cite{c3}. It is useful in applications such as video streaming, voice over IP, and  real-time applications such as networked control and cyberphysical systems where delay is critical.\\
The problem above is characterized by a period $\tau$, and a channel reliability $p_{i}$ and a timely throughput requirement $q_{i}$ for each flow $i=1,\ldots,N$. The first question that arises is: Can the access point meet the timely throughput requirements of each client? This question is answered in~\cite{c3}, where the set of timely throughput vectors $\left(q_{1},\ldots,q_{N}\right)$ that are feasible is characterized.\\
To describe the characterization, let $\gamma_{i}$ denote the geometrically distributed random variable with parameter $p_{i}$ that denotes the number of attempts needed to deliver a packet of client $i$. If $\sum_{i=1}^{N} \gamma_{i}<\tau$, then that implies that all the packets that arrived at the beginning of a frame have been delivered, and so the remaining time in the frame $\left[\tau-\sum_{i=1}^{N}\gamma_{i}\right]$ is idle time where the access point has no more packets to transmit. The quantity 
\begin{align}\label{idle}
I_{\left\{1,\ldots,N\right\}} : = \frac{1}{\tau} E\left(\left[\tau-\sum_{i=1}^{N} \gamma_{i}\right]^{+} \right),
\end{align} 
where $x^{+} := \max\left(x,0\right)$, denotes the expected proportion of time in a frame that the access point is idle when the set of clients is $\left\{1,\ldots,N\right\}$.\\
Hou et al~\cite{c3} have shown that a vector $\left(q_{1},\ldots,q_{N}\right)$ of timely throughputs is feasible if and only if
\begin{align}\label{feasible}
\sum_{i\in S}\frac{q_{i}}{\tau p_{i}}\leq  1 - I_{S}, \mbox{ for every } S\subseteq \left\{1,\ldots,N\right\},
\end{align}
where $I_{S} = \frac{1}{\tau} E\left(\left[\tau-\sum_{i\in S} \gamma_{i}\right]^{+} \right)$ is the proportion of idle time when the set of clients is $S$. We may call this the ``rate region".\\
It is also shown in~\cite{c3} that the following class of weighted-debt policies will satisfy any set of clients for which the timely throughput vector is feasible. Let
\begin{align}
d_{i}\left(t\right) :=& \frac{1}{\alpha_{i}}\Big[tq_{i} -\mbox{number of packets delivered for client  }\notag\\
&\qquad \qquad\qquad \mbox{ $i$ in the frames $\left\{0,\ldots,t-1\right\} \Big]$}
\end{align}  
denote the weighted packet-debt owed to client $i$ at the beginning of frame $t$, where $\alpha_{i}>0$ is a weighting parameter. They showed that if the clients are ordered in terms of decreasing weighted debt at the beginning of a frame, and then served in that order during the frame, then the resulting weighted debt policy is feasibility optimal. By this is meant that if a set of clients can be satisfied, i.e., satisfies~(\ref{feasible}), then this policy will meet their timely throughput requirements.\\
The timely throughput of client $i$, defined as 
\begin{align*}
\liminf_{T\to\infty}\frac{\mbox{Number of packets of user delivered till frame $T$}}{T}, 
\end{align*}
however only captures the long term average. It could be that there are long runs of frames where packets are not delivered. However such fluctuations in short term behaviour are ignored by the long term time average used in the definition of timely throughput. That is, the timely throughput does not capture a finite grained sense of performance. This is similar to queueing networks where there may be many policies that are throughput optimal, but whose delays may be exceedingly large.\\
Motivated by this, we address the issue of performance of real-time scheduling policies in this paper. In this study, we suppose that the system is in heavy traffic. In our context that corresponds to the vector of timely throughput requirements lying on the boundary of the rate region.\\
We assume the same framework as in~\cite{c3}, and study a scaled version of the unweighted debt process where $\alpha_{i}\equiv 1$, i.e.,
\begin{align}
d_{i}\left(t\right) &:= tq_{i} -\mbox{number of packets delivered for client $i$} \notag\\
&\qquad \qquad \mbox{in the frames $\left\{0,\ldots,t-1\right\}$}.
\end{align}
We scale the debt process by $\phi\left(t\right):=\sqrt{2t\log\log t}$ and study the almost sure limit set of the scaled process $\frac{d_{i}\left(t\right)}{\sqrt{2t\log\log t}}$. This is akin to the quantity studied in the law of the iterated logarithm~\cite{c5}. We determine bounds on the limit set.  \\
This is analogous to studying the ``workload process" in queueing systems in the heavy traffic limit. The difference is that in our case the system is taken to be exactly in heavy traffic, and no limit need to be taken. Also, our characterization by the law of the iterated logarithm is a precise sample path characterization of the performance in heavy traffic.\\
The optimality of throughput-optimal policies has been considered in the ``heavy trafic regime" in~\cite{c1,c2}. Our problem differs from the general switch considered there in that the debt processes can become negative. Moreover the approach differs in that a diffusive scaled version of the ``workload process" is studied and the results are given for a compact time interval. Another approach to show optimality of the MaxWeight scheduler under heavy traffic is taken in~\cite{c2}, where the expected value of steady state queue length, multiplied by the distance of the vector of arrival rates from the hyperplanes is studied and MaxWeight is shown to be optimal. Our approach differs from the previous works in that we consider the throughput requirements to lie on the boundary of the rate region and give results for a scaled version of the debt process  as limit $t\to\infty$.

It has been shown in~\cite{c3}, that each of the $2^{N}$ subsets of clients imposes a ``workload, idle-time" constraint, i.e. the sum of workloads over users in a subset must be less than the total time available to that subset of clients, which is the total time minus time spent idling. 
 Our results proceed by showing that the debt vector evolves along the normal to the hyperplane on which the $\bm{q}$ vector lies. Then we derive an upper bound on the sums of the scaled debt processes, which must be satisfied by any non-idling throughput-optimal policy. These two together give us the bound on individual debt processes.
We formulate the problem in Section~\ref{sec2}. Section~\ref{sec3} contains some preliminary results which will be used in later sections. Remaining sections contain results for various cases. In Section~\ref{sec6} we introduce a notion of optimality and associate a cost with a policy. We show that the MWDF policy is optimal with respect to this cost. Finally Section~\ref{sec:conclusion} summarizes the results. 
\section{Problem Formulation}\label{sec2}
For the remainder of this section, we will assume that $\tau = 1$, i.e. frame-length is one time slot. We consider the case where there are two users with $q_{i}, p_{i}, i = 1,2$ as the parameters. We will analyze the scheduling policy that compares $\frac{d_{1}(t)}{p_{1}}$ and $\frac{d_{2}(t)}{p_{2}}$ and serves the client with higher weighted debt, and refer to it as the Maximum Weighted Debt Policy (MWDF). 
First we will show that, 
\begin{align}\label{eq:pf2}
\lim_{t\to\infty}  \frac{1}{\phi(t)}\left|\frac{d_{1}(t)}{p_{1}}-\frac{d_{2}(t)}{p_{2}}\right| = 0\mbox{ a.s.},
\end{align}
where $\phi(t) = \sqrt{2t\log\log t}  $.
To this end we introduce the following lemma taken from~\cite{c2}.
\begin{lemma}\label{lemma1}
For an irreducible and aperiodic Markov Chain $\{X[t]\}_{t\geq 0}$ over a countable state space $\mathcal{X}$, suppose $Z:\mathcal{X} \to \mathbb{R}_{+} $ is a non-negative valued function. Define drift of $Z$ at $X$ as
\begin{align}
\Delta Z(X) \triangleq [	Z(X[t+1])	- Z(X[t])	] \mathbbm{1}\{ X[t] = X\}.
\end{align}
Let the drift satisfy the following conditions:
\begin{itemize}
\item Condition $\mathcal{C}_{1}$: There exists an $\eta > 0$, and a $\kappa <\infty$ such that 
\begin{align*}
E\left[\Delta Z(X) | X[t] = X\right] \leq -\eta, \forall X\in \mathcal{X} \mbox{ with } Z(X) \geq \kappa.
\end{align*}
\item Condition $\mathcal{C}_{2}$: There exists a $D<\infty$ such that 
\begin{align*}
\Pr(|\Delta Z(X)|\leq D)=1  \forall X\in \mathcal{X}.
\end{align*}
\end{itemize}
Then there exists a $\theta^{\star}>0$ and a $C^{\star}<\infty$ such that 
\begin{align*}
\limsup_{t\to\infty} E[\exp^{\theta^{\star}Z(X[t])}  ] \leq C^{\star}. 
\end{align*}
Furthermore, if the Markov Chain is assumed to be positive recurrent, then $Z(X[t])$ converges in distribution to a random variable $\bar{Z}$ for which 
\begin{align*}
E[\exp^{\theta^{\star}\bar{Z}} ] \leq C^{\star},
\end{align*}
which implies that all moments of $\bar{Z}$ exist and are finite.
 \end{lemma}
 Applying Lemma~\ref{lemma1} we obtain,
 \begin{lemma}\label{lemma2}
 Under the MWDF policy~(\ref{eq:pf2}) holds true.
\end{lemma} 
\begin{proof} Let $Z(\bm{d})  =| \frac{d_{2} }{p_{2}}- \frac{ d_{1} }{p_{1}}| $. Let $\kappa = 1+ \frac{1}{\min(p_{1},p_{2}) }$. Clearly when $Z(\bm{d}(t)) > \kappa$, the client with higher debt will be served in the time slot. Let $\frac{d_{2}(t) }{p_{2}}- \frac{ d_{1}(t) }{p_{1}} > \kappa$. 
Then, if $\bm{d}(t)$ is such that $\frac{d_{2}(t)}{p_{2}}- \frac{d_{1}(t)}{p_{1}}>\kappa$,
\begin{align*}
E( \Delta Z(\bm{d}) | \bm{d} ) & = p_{2}   (\frac{d_{2}+q_{2}-1  }{p_{2}} - \frac{d_{1}+q_{1}  }{p_{1}})\\
&+ (1-p_{2}) (\frac{d_{2}+q_{2}  }{p_{2}} - \frac{d_{1}+q_{1}  }{p_{1}})- (\frac{d_{2} }{p_{2}} - \frac{d_{1} }{p_{1}})\\
& = \frac{q_{2} }{p_{2}}- \frac{q_{1}}{p_{1}} -1\\
&=-2\frac{q_{1}}{p_{1}}.
\end{align*}
Similarly, if $\frac{d_{1}(t)}{p_{1}}- \frac{d_{2}(t)}{p_{2}}>\kappa$, then $E( \Delta Z(\bm{d}) | \bm{d} ) <-2\frac{q_{2}}{p_{2}}$. Note also that if $\frac{d_{2}(t) }{p_{2}}- \frac{d_{1}(t)}{p_{1}} > \kappa$, then 
\begin{align*}
\frac{d_{2}(t+1) }{p_{2}}- \frac{d_{1}(t+1)}{p_{1}} &= \frac{d_{2}(t) }{p_{2}}- \frac{d_{1}(t)}{p_{1} }+ \frac{q_{2}}{p_{2}} -\frac{q_{1}}{p_{1}} \\
&-\frac{\mathbbm{1}\{\mbox{packet for $2$ is delivered}\} }{p_{2}}\\
&\geq \frac{d_{2}(t) }{p_{2}}- \frac{d_{1}(t)}{p_{1} } -1 - \frac{1}{p_{2}}\\
&\geq 0,
\end{align*}
where we have used the fact that $0<\frac{q_{i}}{p_{i}} <1$ and $\frac{d_{2}(t) }{p_{2}}- \frac{d_{1}(t)}{p_{1} } >\kappa$. Hence the clients maintain the order of debts if the difference is greater than $\kappa$. 
So we can take $\eta = -2\min{\frac{q_{1}}{p_{1}}, \frac{q_{2}}{p_{2}}}$.It remains to show that the drifts are bounded for all $\bm{d}$. If $ \frac{d_{2} }{p_{2}}- \frac{d_{1}}{p_{1}} > \kappa$ then since the debts don't change the order, the drift is bounded simply by 
\begin{align*}
\frac{q_{1}}{p_{1}} + \frac{q_{2}}{p_{2}} + \frac{1}{p_{1}} +\frac{1}{p_{2}}.
\end{align*}
On the other hand, if $ \frac{d_{2} }{p_{2}}- \frac{d_{1}}{p_{1}} < \kappa$, then
\begin{align*}
&\left| \frac{d_{2}(t+1)}{p_{2}}	- \frac{d_{1}(t+1)}{p_{1}}\right| - \left| \frac{d_{2}(t)}{p_{2}}	- \frac{d_{1}(t)}{p_{1}}\right| \\
\leq &     \Bigg| \frac{d_{2}(t+1)}{p_{2}}-\frac{d_{2}(t)}{p_{2}}+\frac{d_{2}(t)}{p_{2}}\\
&\qquad+\frac{d_{1}(t)}{p_{1}}	- \frac{d_{1}(t+1)}{p_{1}}-\frac{d_{1}(t)}{p_{1}}\Bigg| + \kappa \\
\leq& \left| \frac{d_{2}(t+1)}{p_{2}}-\frac{d_{2}(t)}{p_{2}} \right| + \left|\frac{d_{1}(t)}{p_{1}}	- \frac{d_{1}(t+1)}{p_{1}}\right|\\
&\qquad+ \left| \frac{d_{2}(t)}{p_{2}}-\frac{d_{1}(t)}{p_{1}} \right| +\kappa \\
\leq &\frac{1+q_{2}}{p_{2}} + \frac{1+q_{1}}{p_{1}} +2\kappa.
\end{align*}
Hence for some finite $N(\epsilon)$, and $\theta^{\star} >0$, 
\begin{align}\label{eq:pf3}
 E[\exp^{\theta^{\star} |\frac{d_{2}(t)}{p_{2}} - \frac{d_{1}(t)}{p_{1}}|}]< C^{\star} , \forall t>N(\epsilon). 
\end{align}

Now applying the Borel Cantelli lemma, we have for any $\delta>0$,
\begin{align}\label{eq:pf4}
&\sum_{t=1}^{\infty} \Pr\left(  \frac{1}{\phi(t)}\Big|\frac{d_{1}(t)}{p_{1}}-\frac{d_{2}(t)}{p_{2}}\Big|  >\delta     \right)  \notag\\
\leq &\frac{BN(\epsilon)}{\delta} + \frac{C^{\star} + \epsilon}{\exp(\theta^{\star} \delta\phi(t) )}  \notag\\
\leq &\frac{BN(\epsilon)}{\delta} +\frac{C^{\star} + \epsilon}{\exp(\delta \theta^{\star}\sqrt{t}) }\notag\\
<&\infty,
\end{align}
where the first inequality follows from inequality~(\ref{eq:pf3}) applied with Markov inequality, and the second inequality follows by the fact that $\phi\left(t\right) > \sqrt{t}$. Now~(\ref{eq:pf4}) completes the proof. 
\end{proof}
\section{Preliminary Results}\label{sec3}
Now we will introduce some auxiliary random variables and results, which will be useful in the later proofs. 
Let 
\begin{align}\label{def:martin}
m_{j}(t) &= \mathbbm{1} \{\mbox{user $j$ is attempted at time $t$  }\} \notag\\
& \quad- \frac{1}{p_{j}} \mathbbm{1} \{\mbox{user $j$ is delivered at time $t$  }\} \notag \\
&=u_{j}(t)  - \frac{1}{p_{j}} g_{j}(t) .
\end{align}
Clearly 
\begin{align}\label{def:martin2}
M_{j}(t) = \sum_{l=1}^{t} m_{j}(l),
\end{align}
is a martingale since
\begin{align*}
& E[M_{j}(t+1)| \mathcal{F}_{t} ] = E[M_{j}(t) + m_{j}(t+1) | \mathcal{F}_{t}]\\
=& M_{j}(t) + E\left[u_{j}(t+1)  - \frac{1}{p_{j}} g_{j}(t+1)| \mathcal{F}_{t}  \right]\\
=& M_{j}(t)+ E\bigg[ E\Big[u_{j}(t+1)  \\ 
&\qquad\qquad- \frac{1}{p_{j}} g_{j}(t+1)\Big| \mathcal{F}_{t},u_{j}(t+1)\Big]\bigg | \mathcal{F}_{t}  \bigg]\\
=& M_{j}(t) 
\end{align*}
We will now provide some limit results for any policy for which $q_{1},q_{2}$ is feasible. First we state the law of large numbers taken from~\cite{c6}\\
\begin{theorem}\label{llnmartingale}
Let $Y_{n} = \sum_{i=1}^{n} X_{i}$ be a Martingale such that $\sum_{n=1}^{\infty} \frac{E(X^{2}_{n})}{n^{2}} <\infty$. Then $\frac{Y_{n}}{n}\to 0$ a.s.
\end{theorem}
This immediately gives us the following for any throughput optimal policy.
\begin{lemma}\label{lemma2.5}
The following is true for any non-idling policy fulfilling $q_{1}, q_{2}$:
\begin{align*}
\lim_{T\to\infty} \frac{\sum_{t=1}^{T} u_{j}(t)  }{T}=\frac{q_{j}}{p_{j}}.
\end{align*}
\end{lemma}
\begin{proof}
 Clearly $E(m_{j}^{2} (t)  ) \leq (1+\frac{1}{p_{j}})^{2}$ and so the conditions of Theorem~(\ref{llnmartingale}) are satisfied. Hence,
\begin{align}
\frac{M_{j}(t)}{t} \to 0,
\end{align}
and so
\begin{align}
 \frac{\sum_{l=1}^{t}u_{j}(l)}{t} \to \frac{1}{p_{j}}\frac{\sum_{l=1}^{t}s_{j}(l)}{t}.
\end{align}
But since $\liminf_{t\to\infty} \frac{\sum_{l=1}^{t}s_{j}(l)}{t} \geq q_{j}$, we have 
\begin{align}\label{eq:100}
\liminf_{t\to\infty} \frac{\sum_{l=1}^{t}u_{j}(l)}{t} \geq \frac{q_{j} }{p_{j}}.
\end{align}
Now, 
\begin{align*}
\sum_{j=1}^{N}   \frac{\sum_{l=1}^{t}u_{j}(l)}{t} = 1,
\end{align*}
holds for all $t$. If $\limsup_{t\to\infty}   \frac{\sum_{l=1}^{t}u_{j}(l)}{t} >\frac{q_{j}}{p_{j}}$, then this would imply $\liminf_{t\to\infty}   \frac{\sum_{l=1}^{t}u_{k}(l)}{t} <\frac{q_{k}}{p_{k}}$ for some $k$ since $\sum_{j=1}^{N}\frac{q_{j}}{p_{j}} = 1$, which contradicts~(\ref{eq:100}). This completes the proof.
\end{proof}
Let $(X_{n},\mathcal{F}_{n})_{n\geq 0}$ be a Martingale with $Y_{n} = X_{n} - X_{n-1}$ and $s^{2}_{n} =  \sum_{i=1}^{n} E[Y^{2}_{i} | \mathcal{F}_{i-1}     ]$. Also let $u_{n} = \sqrt{2n\log\log n }$

The following is the law of iterated logarithm for martingales~\cite{c4}.
\begin{theorem}
\label{theorem:stout}
If $s_{n}^{2} \to \infty$ and $|Y_{n}| \leq \frac{ K_{n}  s_{n} }{u_{n}}$, where $K_{n}$ are $\mathcal{F}_{n-1}$ measurable with $K_{n} \to 0$, then 
\begin{align}\label{eq:pf6}
\limsup_{n\to\infty}  \frac{X_{n}}{s_{n} u_{ n}} = 1.
\end{align}
\end{theorem}
 The following lemma uses Theorem~(\ref{theorem:stout}) and will be crucial in later results. Equalities are to be understood in almost sure sense.
 \begin{lemma}\label{lemma3}
 Consider the Martingale $M_{j}(t)$. The following are true for any non-idling policy that fulfils $q_{1},q_{2}$:
 \begin{enumerate}[(i).]
 \item $E(m_{j}^{2}(t) |\mathcal{F}_{t-1} ) = u_{j}(t)  v_{j} $, with $v_{j} = \frac{(1-p_{j})}{p_{j}}$. 
 \item $\mathcal{S}^{2}_{j}(T) \triangleq \sum_{j=1}^{T}E(m_{j}^{2}(t) |\mathcal{F}_{t-1} ) = \sum_{t=1}^{T}u_{j}(t)  v_{j} $.
 \item $\lim_{T\to\infty}\frac{\mathcal{S}^{2}_{j}(T) }{T}  = \frac{q_{j} }{p_{j}} v_{j}  = \frac{q_{j} (1-p_{j})}{p_{j}^{2}}\triangleq c_{j} $.
 \item $\lim_{T\to\infty} \frac{\mathcal{S}_{j}(T) }{\sqrt{T}}  = \sqrt{c_{j}}  $
 \end{enumerate}
 Using the above three, also noting that $m_{j}(t)$ is bounded for all $t$, and appliying Theorem~\ref{theorem:stout} with, $K_{n} = \frac{\sqrt{\log\log n }} { \sum_{t=1}^{n}E(m_{j}^{2}(t) |\mathcal{F}_{t-1} )} \to 0	$ gives us,
 \begin{align}\label{eq:pf7}
 \limsup_{T\to\infty} \frac{M_{j}(t)}{\phi(t)} = \sqrt{c_{j}}.
 \end{align}
 \end{lemma}

\begin{theorem}\label{theorem:main}
Under the MWDF policy, 
\begin{align}
\limsup_{t\to\infty}\frac{d_{j}(t)}{\phi(t)} \leq p_{j} \frac{\sqrt{c_{1}} + \sqrt{c_{2}}}{2} .
\end{align}
\end{theorem}

\begin{proof} By~(\ref{def:martin}) and~(\ref{def:martin2}), we have,
\begin{align}
M_{1}(t) + M_{2}(t) &= t - \frac{s_{1} (t)}{p_{1}} -\frac{s_{2} (t) }{p_{2}}\notag\\
&= t(\frac{q_{1}}{p_{1}}  + \frac{q_{2}}{p_{2}})  - \frac{s_{1} (t)}{p_{1}} -\frac{s_{2} (t) }{p_{2}}\notag\\
&= \frac{d_{1} (t)}{p_{1}} + \frac{d_{2} (t)}{p_{2}}.
\end{align}
Now,
\begin{align*}
\begin{array}{rr@{\hspace{0em}}c@{\hspace{0em}}l}
& 2\frac{d_{1} (t)}{p_{1}} + \frac{d_{2} (t)}{p_{2}}-\frac{d_{1} (t)}{p_{1}} &=& M_{1}(t) + M_{2}(t) \notag \\
\implies &\frac{1}{\phi(t)}   \left[2\frac{d_{1} (t)}{p_{1}} + \frac{d_{2} (t)}{p_{2}}-\frac{d_{1} (t)}{p_{1}} \right] &=&\frac{M_{1}(t) }{\phi(t)}+ \frac{M_{2}(t)}{\phi(t)}\\
\implies &\frac{2}{\phi(t)}   \frac{d_{1} (t)}{p_{1}}  &=&\frac{M_{1}(t) }{\phi(t)}+ \frac{M_{2}(t)}{\phi(t)} \\
&&& -\left[  \frac{d_{2} (t)}{p_{2}}-\frac{d_{1} (t)}{p_{1}}    \right]\\
\end{array}
\end{align*}
\begin{align*}
\implies \limsup_{t\to\infty}\frac{2}{\phi(t)}   \frac{d_{1} (t)}{p_{1}}  &\leq \limsup_{t\to\infty}\frac{M_{1}(t) }{\phi(t)}\\
&+ \limsup_{t\to\infty}\frac{M_{2}(t)}{\phi(t)} \\
&+\limsup_{t\to\infty} \left[  \frac{d_{2} (t)}{p_{2}}-\frac{d_{1} (t)}{p_{1}}    \right].
\end{align*}
The third of the above quantities on right side is $0$ by Lemma~\ref{lemma2}. So applying Lemma~\ref{lemma3} we get the required result. 
\end{proof}

 \section{Two Users With General Frame length}\label{sec4}
We will look at the case when the timely throughputs lie on exactly one hyperplane, given by
\begin{align}\label{system1}
\frac{q_{1}}{p_{1}} +\frac{q_{2}}{p_{2}} = \tau(1-I_{1,2}).
\end{align}
where $I_{1,2}$ is the expected idle time in a frame after delivering both packets.
The MWDF policy compares $\frac{d_{1} (t)}{p_{1}}$ and $\frac{d_{1} (t)}{p_{1}}$ and serves them in the decreasing order. As before we will consider a Lyapunov function to show~(\ref{eq:pf2}) .
\begin{lemma}\label{lemma8}
For the system~(\ref{system1}) under the MWDF policy~(\ref{eq:pf2}) holds true.
\end{lemma}
\begin{proof} Let $Z(\bm{d})  =| \frac{d_{2} }{p_{2}}- \frac{ d_{1} }{p_{1}}| $. Let $\kappa = 1+ \frac{2}{\min(p_{1},p_{2})}$. Let $\frac{d_{2}(t) }{p_{2}}- \frac{ d_{1}(t) }{p_{1}} > \kappa$. Then, the order $\{2,1\}$ is used in frame $t+1$. Let $\pi_{2}$ and $\pi_{1}$ denote the probabilities that packets for respective users are delivered under the ordering $\{2,1\}$. It has been shown in~\cite{c3} that, 
\begin{align*}
\frac{\pi_{2}}{p_{2}} &= \tau(1-I_{\left\{2\right\}}),\\
\frac{\pi_{2}}{p_{2}} + \frac{\pi_{1}}{p_{1}}& = \tau(1-I_{\left\{1,2\right\}}).\\
\end{align*} 
For $\bm{d}$ such that $\frac{d_{2}(t)}{p_{2}} - \frac{d_{1}(t)}{p_{1}}>\kappa$, we have
\begin{align*}
&E\left( \Delta Z\left(\bm{d}\right) | \bm{d}  \right) \\
=& E\left[ \frac{d_{2}(t+1) }{p_{2}}- \frac{d_{1}(t+1)}{p_{1}} -\left(\frac{d_{2}(t) }{p_{2}}- \frac{d_{1}(t)}{p_{1}} \right)\Bigg| \mbox{order }\{2,1\}  \right]\\
=&  E\Big[ \left(\frac{d_{2}(t+1) }{p_{2}}-\frac{d_{2}(t) }{p_{2}}\right) \\
&\qquad- \left(\frac{d_{1}(t+1)}{p_{1}}- \frac{d_{1}(t)}{p_{1}} \right)\Big|\mbox{order } \{2,1\}  \Big]\\
=& \frac{q_{2}}{p_{2}} - \frac{\pi_{2}}{p_{2}} -\left(\frac{q_{1}}{p_{1}} - \frac{\pi_{1}}{p_{1}} \right)\\
=& \frac{q_{2}}{p_{2}} -2\frac{\pi_{2}}{p_{2}}-\frac{q_{1}}{p_{1}}+\frac{\pi_{1}}{p_{1}}+\frac{\pi_{2}}{p_{2}}\\
=&\frac{q_{2}}{p_{2}}-2\tau\left( 1-I_{\left\{2\right\}}\right)-\frac{q_{1}}{p_{1}} + \tau\left(1-I_{\left\{1,2\right\}}\right)\\
=&2\frac{q_{2}}{p_{2}}-2\tau\left( 1-I_{\left\{2\right\}}\right)-\left(\frac{q_{1}}{p_{1}}+\frac{q_{2}}{p_{2}}\right) + \tau\left(1-I_{\left\{1,2\right\}}\right)\\
=& 2\frac{q_{2}}{p_{2}}-2\tau\left( 1-I_{\left\{2\right\}}\right)-\tau\left(1-I_{\left\{1,2\right\}}\right) + \tau\left(1-I_{\left\{1,2\right\}}\right)\\
=&2\left( \frac{q_{2}}{p_{2}} - \tau\left(1-I_{\left\{2\right\}}\right)   \right)\\
<&0.
\end{align*}
Similarly 
\begin{align*}
E( \Delta Z(\bm{d}) | \bm{d}  )<0,
\end{align*}
 when $\frac{d_{1}(t)}{p_{1}} - \frac{d_{2}(t)}{p_{2}}>\kappa$. 
So we can take $\eta = \max{2\left( \frac{q_{2}}{p_{2}} - \left(1-I_{\left\{2\right\}}\right)   \right), 2\left( \frac{q_{1}}{p_{1}} - \left(1-I_{\left\{1\right\}}\right)   \right)}$.\\ 
We need to show that the drift is uniformly bounded. As in previous section, once again if $ \frac{d_{2} }{p_{2}}- \frac{d_{1}}{p_{1}} > \kappa$, drift is bounded simply by 
\begin{align*}
\frac{1+q_{1}}{p_{1}} + \frac{1+q_{2}}{p_{2}} .
\end{align*}
If $ \frac{d_{2} }{p_{2}}- \frac{d_{1}}{p_{1}} < \kappa$, then once again,
\begin{align*}
&\left| \frac{d_{2}(t+1)}{p_{2}}	- \frac{d_{1}(t+1)}{p_{1}}\right| - \left| \frac{d_{2}(t)}{p_{2}}	- \frac{d_{1}(t)}{p_{1}}\right| \\
&\leq \frac{1+q_{2}}{p_{2}} + \frac{1+q_{1}}{p_{1}} +2\kappa
\end{align*}
For the remaining part, we can once again apply Borel Cantelli to get the claim.
\end{proof}
Let $i(t)$ denote the time spent idling in frame $t$. Note that $i(t)$ are independent and identically distributed with mean $\tau I_{\left\{1,2\right\}}$ and finite variance, which will be denoted as $\sigma^{2}_{I}$. Next we state without proof the results similar to previous section. The derivation of the first result is exactly the same as in the previous section, while the second follows using Kolmogorov's Law of the Iterated Logarithm.
\begin{lemma}
\begin{enumerate}
\item 
\begin{align*}
 \limsup_{T\to\infty} \frac{M_{j}(t)}{\phi(t)} = \sqrt{c_{j}}
\end{align*},
\item 
\begin{align*}
\limsup_{t\to\infty} \frac{\sum_{l=1}^{t} \left(i(l) -\tau I_{1,2}\right)}{\phi(t)} = \sigma_{I}.
\end{align*}
\end{enumerate}

\end{lemma} 
\begin{theorem}\label{theorem:main1}
Under the MDF policy, 
\begin{align}
\limsup_{t\to\infty}\frac{d_{j}(t)}{\phi(t)} \leq p_{j} \frac{\sqrt{c_{1}} + \sqrt{c_{2}} + \sigma_{I} }{2} .
\end{align}
\end{theorem}
\section{Multiple Users With $\tau $=1}\label{sec5}
In this section we consider the case of $N$ users with frame size of one time-slot.
Since there is no idling in this case, the rate region is given by,
\begin{align}\label{rateregion}
\sum_{j=1}^{N}\frac{q_{j}}{p_{j}} \leq 1.
\end{align}
For the remainder of this section $\bm{q}$ will denote the vector $(q_{1},\ldots,q_{N}  )$, and it will be assumed to lie on the hyperplane defined by equation in~(\ref{rateregion}). First we will show that, 
\begin{align}\label{eq:sec3:1}
\lim_{t\to\infty}  \frac{1}{\phi(t)}\left|\frac{d_{j}(t)}{p_{j}}-\frac{d_{k}(t)}{p_{k}}\right| = 0\mbox{ a.s.},\forall j,k= 1,\ldots,N. 
\end{align}
We introduce the following Lyapunov function,
\begin{align}\label{sec:3:lyapunov}
Z(\bm{x}) = \sum_{j=1}^{N}(x_{j}-x_{\min}),
\end{align}
where $x_{\min}$ is the minimum entry of $\bm{x}$. Let $c: = 1+ \max_{j=1,\ldots,N} \frac{1+2q_{j}}{p_{j}} $.
\begin{lemma}\label{sec:3:lemma1}
 Under the MWDF policy~(\ref{eq:sec3:1}) holds true.
\end{lemma} 
\begin{proof}
 Consider $Z(\bm{x}) $ as in~(\ref{sec:3:lyapunov}). We will use the notation $\bm{\frac{d(t)}{p}}$ to refer to the vector with entries $\frac{d_{j}(t)}{p_{j}}$. Note that
\begin{align*}
Z(\bm{x})>Nc &\implies \sum_{j=1}^{N}(x_{j}-x_{\min}) >Nc\\
&\implies N(x_{\max}-x_{\min}) >Nc\\
&\implies (x_{\max}-x_{\min}) >c.
\end{align*}
Suppose now that 
\begin{align*}
Z\left(\bm{\frac{d(t)}{p}}\right)>Nc.
\end{align*}
Then,
\begin{align*}
\left(\bm{\frac{d(t)}{p}}\right)_{\max} >c.
\end{align*}
Let $j\in \arg\max \frac{d_{j}(t) }{p_{j}} $. Then,
\begin{align*}
\frac{d_{j}(t+1) }{p_{j}} -\left[\bm{\frac{d(t+1)}{p}}\right]_{\min}&>   \frac{d_{j}(t) }{p_{j}} - \frac{1+q_{j}}{p_{j}}\\
& - \left[\left(\bm{\frac{d(t)}{p}}\right)_{\min} +\max_{j=1,\ldots,N}\frac{q_{j}}{p_{j}} \right]\\
&>0,
\end{align*}
where the inequalities are the result of the facts that the client being served cannot decrease in debt by more than $-1+q_{j}$, and the maximum increment in the minimum value can be $\max_{j=1,\ldots,N}\frac{q_{j}}{p_{j}} $.
Set $\kappa = Nc$. Also note that since the frame size is one slot, only one user $j\in \arg\max \frac{d_{j}(t) }{p_{j}} $ will be receiving service in the present frame and hence the debts of users at time $t+1$, other than this maximum weighted debt user, are known at time $t$ itself. Also if $Z\left(\bm{\frac{d(t)}{p}}\right)>Nc$ we know that the max weighted debt user will not be the minimum weighted debt user in the next frame. We will use $Z_{t+1}$ to mean $Z(\frac{\bm{d}(t+1)}{\bm{p}})$ and $j_{t},j_{t+1}$ to refer to the indices of the clients with the least weighted debt at times $t$ and $t+1$ respectively, and by $j^{\max} $ to the index of the client with highest weighted debt at the beginning of frame $t+1$. This gives us, for $Z_{t}>Nc$,
\begin{align}\label{sec:3:lypn1}
& E\left[ Z_{t+1} |   Z_{t}   \right] = \sum_{j\neq j_{t+1},j^{\max}} \left(\frac{d_{j}(t)}{p_{j}} + \frac{q_{j}}{p_{j}}\right) \notag\\
&-\left(\frac{d_{j_{t+1}}(t)}{p_{j_{t+1}}} + \frac{q_{j_{t+1}}}{p_{j_{t+1}}}\right)+ \left(\frac{d_{j^{\max}}(t)}{p_{j^{\max}}} + \frac{q_{j^{\max}}}{p_{j^{\max}}} -1\right)  \notag\\
&\qquad -\left(\frac{d_{j_{t+1}}(t)}{p_{j_{t+1}}} + \frac{q_{j_{t+1}}}{p_{j_{t+1}}}\right)
\end{align}
Also,
\begin{align}\label{sec:3:lypn2}
 & Z_{t} = \sum_{j\neq j_{t},j^{\max}} \left(\frac{d_{j}(t)}{p_{j}} \right)  -\left(\frac{d_{j_{t}}(t)}{p_{j_{t}}} \right)\notag\\
&\qquad+ \left(\frac{d_{j^{\max}}(t)}{p_{j^{\max}}} \right)  -\left(\frac{d_{j_{t}}(t)}{p_{j_{t}}}\right).
\end{align}
Subtracting the above two, and using~(\ref{rateregion}) we get for $d(t) = d$ such that  $Z\left(\frac{\bm{d}(t)}{\bm{p}}  \right)>Nc$,
\begin{align}\label{sec3:lema11}
& E\left[\Delta Z(\bm{d})| \bm{d} \right]\notag\\
=& N\left[  \frac{d_{j_{t}}(t)}{p_{j_{t}}} - \frac{d_{j_{t+1}}(t)}{p_{j_{t+1}}} - \frac{q_{j_{t+1}}}{p_{j_{t+1}}}  \right]\notag\\
<& 0.
\end{align}
In the last inequality we have used the fact,
\begin{align}\label{fact1}
\frac{d_{j_{t+1}}(t)}{p_{j_{t+1}}} \geq  \frac{d_{j_{t}}(t)}{p_{j_{t}}}
\end{align}
and also that,
\begin{align}
 \frac{d_{j_{t}}(t)}{p_{j_{t}}} + \frac{q_{j_{t}}}{p_{j_{t}}} \geq  \frac{d_{j_{t+1}}(t)}{p_{j_{t}+1}} + \frac{q_{j_{t+1}}}{p_{j_{t+1}}},
 \end{align}
and hence
 \begin{align}\label{fact2}
 \frac{d_{j_{t}}(t)}{p_{j_{t}}} -\frac{d_{j_{t+1}}(t)}{p_{j_{t}+1}} \geq  \frac{q_{j_{t+1}}}{p_{j_{t+1}}} - \frac{q_{j_{t}}}{p_{j_{t}}}.
  \end{align}
Combining~(\ref{fact1}) and~(\ref{fact2}), when $Z\left(\frac{\bm{d}(t)}{\bm{p}}\right)>Nc$  
\begin{align}
-N\left(\frac{q}{p}\right)_{\max}\leq E\left[\Delta Z(\bm{d})| \bm{d}\right]\leq -N\left(\frac{q}{p}\right)_{\min}.
\end{align}
The boundedness of drift when $Z(\bm{d})\leq Nc$ follows trivially by boundedness of the increments of individual debts by $\frac{q_{j}+1}{p_{j}}$. Hence we can use Lemma~\ref{lemma1} in conjunction with the Borel Cantelli lemma to infer that
\begin{align*}
\lim_{t\to\infty}  \frac{1}{\phi(t)} \left[\sum_{j=1}^{N}\frac{d_{j}(t)}{p_{j}}-\left(\bm{\frac{d(t)}{p}}\right)_{\min} \right]= 0\mbox{ a.s.}
\end{align*}
Therefore,
\begin{align*}
\lim_{t\to\infty}  \frac{1}{\phi(t)} \left[\frac{d_{j}(t)}{p_{j}}-\frac{d_{k}(t)}{p_{k}}\right] = 0\mbox{ a.s.},\forall j,k.
\end{align*}
\end{proof}
The main part of the rest of the analysis lies in showing that the debt vector evolves along the perpendicular to the hyperplane. The rest of the analysis is exactly along the lines of section~\ref{sec2}, and is omitted. We directly state the main result.
\begin{theorem}\label{theorem:main2}
Under the MWDF policy, 
\begin{align}
\limsup_{t\to\infty}\frac{d_{j}(t)}{\phi(t)} \leq p_{j} \frac{\sum_{j=1}^{N}\sqrt{c_{j}}}{N} .
\end{align}
\end{theorem}
\section{General Case With Symmetric Users}\label{sec6}
We will consider the case where there are $N$ users with the vector of timely throughput requirements denoted as $\bm{q}$ with all entries as $q$, i.e. $q_{i} \equiv q,\forall i$. The channel reliabilities of each user will be the same and denoted as $p$, i.e. $p_{i} \equiv p,\quad\forall i$. The frame size is $\tau$ time slots. The rate-region is characterized as the intersection of the following $2^{N}$ halfspaces:
\begin{align}\label{rateregion1}
\sum_{j\in S} \frac{q_{j}}{p_{j}} \leq 1-I_{S}, S\subseteq \{1,\ldots,N\}.
\end{align}
We wil assume that $\bm{q}$ lies on the hyperplane 
\begin{align}\label{hyper}
N\frac{q}{p} = \tau\left(1-I_{\left\{1,\ldots,N\right\}}\right),
\end{align}
and specifically in its relative interior.
The following holds for any non-idling policy,
\begin{align}\label{debtevol}
&\sum_{j=1}^{N} \frac{d_{j}\left(t+1\right) - d_{j}\left(t\right)}{p}\notag\\
&= \sum_{j=1}^{N}  \frac{q -\mathbbm{1}\{\mbox{packet of user $j$ is delivered in frame $t+1$ } \}  }{p}\notag\\
&= \sum_{j=1}^{N} \frac{q}{p} \notag\\
&\qquad-\sum_{j=1}^{N}\frac{\mathbbm{1}\{\mbox{packet of user $j$ is delivered in frame $t+1$ } \}  }{p}.\notag\\
\end{align}
Denoting the probability that a packet of user $j$ is delivered in frame $t+1$ by $\pi_{j}$,
\begin{align}\label{debtevol1}
E\left[\sum_{j=1}^{N} \frac{d_{j}\left(t+1\right) - d_{j}\left(t\right)}{p}\right] &= \sum_{j=1}^{N} \frac{q}{p} - \sum_{j=1}^{N} \frac{\pi_{j}}{p} \notag \\
&=  \sum_{j=1}^{N} \frac{q}{p} - \tau\left(1-I_{\left\{1,\ldots,N\right\}}\right)\notag \\
&=0,
\end{align}
where we have used  $\frac{\pi_{j}}{p} =  \tau\left(1-I_{\left\{1,\ldots,N\right\}}\right)$ from~\cite{c3} page 5 and~(\ref{hyper}). Also note that the distribution of $\sum_{j=1}^{N} \frac{d_{j}\left(t+1\right) - d_{j}\left(t\right)}{p}$ is the same for any non-iding policy, and the distribution of \\
\begin{align*}
&\sum_{j=1}^{N}\mathbbm{1}\{\mbox{packet of user $j$ is delivered in frame $t+1$ } \}\\
&\qquad \triangleq y(t) 
\end{align*}
is given by
\begin{align*}
\Pr\left(y(t) = x\right) = \begin{cases}
{\tau \choose x}p^{x} \left(1-p\right)^{\tau-x} \mbox{ if $x < N$},\\
{ \sum_{y=N}^{\tau}{y \choose N}p^{N} \left(1-p\right)^{y-N} \mbox{ if $x = N$}},\\
{0	\mbox{ if $x > N$}}.
\end{cases}
\end{align*}
Denote by $\sigma_{p,\tau}$ the variance of $\sum_{j=1}^{N} \frac{d_{j}\left(t+1\right) - d_{j}\left(t\right)}{p}$. Summing~(\ref{debtevol}) over $t$, using the fact that $\sum_{j=1}^{N} \frac{d_{j}\left(t+1\right) - d_{j}\left(t\right)}{p}$ are i.i.d. with mean zero~(\ref{debtevol1}) and finite variance, we can use Kolmogorov's law of the iterated logarithm to infer that,
\begin{align}\label{kolmo}
\limsup_{t\to\infty} \frac{1}{\phi\left(t\right)}\left[\sum_{j=1}^{N} \frac{d_{j}\left(t\right)}{p} \right] &= \sigma_{p,\tau} \mbox{ a.s., and }\notag\\
\liminf_{t\to\infty} \frac{1}{\phi\left(t\right)}\left[\sum_{j=1}^{N} \frac{d_{j}\left(t\right)}{p} \right] &= -\sigma_{p,\tau} \mbox{ a.s.}
\end{align}
\begin{lemma}\label{ssc}
For the system~(\ref{hyper}), 
\begin{align}
\lim_{t\to\infty}\frac{1}{\phi\left(t\right)}\left[d_{j}(t)-d_{k}(t)\right] = 0 \mbox{ a.s. }.
\end{align} 
\end{lemma}
\begin{proof}
Note that 
\begin{align*}
d_{j}(t)-d_{k}(t) \leq 1 ,\forall j,k,
\end{align*}
for all $t$ since the MWDF policy gives preference to the client with higher debt.
\end{proof}
\begin{theorem}\label{theoremnotopt}
Under the MWDF policy for the system~(\ref{hyper}), we have
\begin{align}\label{hyperres}
\limsup_{t\to\infty} \frac{1}{\phi\left(t\right)}\frac{d_{j}\left(t\right)}{p} = \frac{\sigma_{p,\tau}}{N} \mbox{ a.s.} ,\forall j.
\end{align}
\end{theorem}
\begin{proof}
This follows by using Lemma~\ref{ssc} in conjunction with~(\ref{kolmo}).  
\end{proof}
\subsection{A Notion of Optimality.}
We now introduce an appropriate notion of optimality for the heavy traffic regime and show that the MWDF policy is optimal for the heavy traffic regime for the symmetric scenario of this section. We will restrict ourselves to policies satisfying $\limsup_{t\to\infty} \frac{1}{\phi\left(t\right)}\left|\frac{d_{j}\left(t\right)}{p} \right| $ is finite a.s.. Also we restrict ourselves to policies such that, 
\begin{align}\label{notionopt}
\limsup_{t\to\infty} \frac{1}{\phi\left(t\right)}\frac{d_{j}\left(t\right)}{p} &= v^{\pi}_{j} \mbox{ a.s. } j = 1,\ldots,N \mbox{ and }\notag\\
\liminf_{t\to\infty} \frac{1}{\phi\left(t\right)}\frac{d_{j}\left(t\right)}{p} &= w^{\pi}_{j}\mbox{ a.s. }  j = 1,\ldots,N \notag\\
\mbox{ such that } \sum_{j=1}^{N} v^{\pi}_{j} &= \sigma_{p,\tau},\notag\\
\qquad \sum_{j=1}^{N} w^{\pi}_{j} &= -\sigma_{p,\tau},
\end{align}
where the superscript denotes the variables corresponding to policy $\pi$. Note that $v^{\pi}_{j}$ and $w^{\pi}_{j}$ can be negative or positive. The set of policies satisfying~(\ref{notionopt}) will be denoted as $\Pi$. It is obvious that we do not lose anything by restricting to policies in $\Pi$ since~(\ref{kolmo}) is true for any non-idling policy. The cost associated with a policy is
\begin{align*}
\max_{j=1,\ldots,N}v^{\pi}_{j}.
\end{align*} 
More precisely, we are interested in the following optimization problem,
\begin{align}\label{cost}
\min_{\pi\in\Pi} &\max_{j=1,\ldots,N} v^{\pi}_{j}\notag\\
\mbox{ such that }& \sum_{j=1}^{N}v^{\pi}_{j} = \sigma_{p,\tau},
\end{align}
where $v_{j}$ are as in~(\ref{notionopt}), and $\pi$ denotes a policy. It is clear that the optimal value of this problem is bounded below by $\frac{\sigma_{p,\tau}}{N}$ and occurs when all $v_{j}$ are equal. Since the MWDF policy attains this bound (Theorem~\ref{theoremnotopt}), it is optimal.\\
Optimality of MWDF policy can also be shown in the cases considered in the previous sections. The performance cost defined as above can be seen as guaranteeing a sort of ``fairness" amongst different flows.
\section{Conclusions and Further Work}\label{sec:conclusion} 
 We have analyzed the performance of real-time wireless networks in heavy traffic. We have performed an analysis of the scaled version of debt process and provided bounds on the process as $t\to\infty$. We have also introduced a ``notion of optimality" and have shown that the MWDF policy is optimal in this sense. We believe an analysis of the general case is possible, the primary step being to show that the debt process evolves along the perpendicular to the hyperplane on which the timely throughput vector lies. A careful construction of the Lyapunov function is to be done to achieve this. Once this is done, it should be possible to show the optimality of the MWDF policy with respect to the performance measure defined in the previous section.

\bibliographystyle{IEEEtran}
\bibliography{IEEEabrv,reference}
\end{document}